\documentclass[12pt,a4paper]{article}

\usepackage{amsmath}
\usepackage{amssymb}
\usepackage{amsthm}
\usepackage{latexsym}
\usepackage[dvipdfmx]{graphicx}
\usepackage{bm}
\usepackage[dvips]{color}

\theoremstyle{plain}
\newtheorem{theorem}{Theorem}
\newtheorem{corollary}[theorem]{Corollary}

\newtheorem{lemma}[theorem]{Lemma}
\newtheorem{claim}{Claim}

\theoremstyle{definition}

\newcommand{\ZZ}{{\mathbb{Z}}}
\newcommand{\RR}{{\mathbb{R}}}

\newcommand{\C}{\mathcal{C}}

\newcommand{\Mdeg}{\mathbf{M}_{\mathrm{in}}}
\newcommand{\Mdegi}[1]{\mathbf{M}_{\mathrm{in}}\sp{({#1})}}
\newcommand{\Idegi}[1]{\mathcal{I}_{\mathrm{in}}\sp{({#1})}}

\newcommand{\MVi}[1]{\mathbf{M}_{V}\sp{({#1})}}
\newcommand{\Mvi}[1]{\mathbf{M}_{v}\sp{({#1})}}
\newcommand{\IVi}[1]{\mathcal{I}_{V}\sp{({#1})}}
\newcommand{\Ivi}[1]{\mathcal{I}_{v}\sp{({#1})}}
\newcommand{\Vi}[1]{V\sp{({#1})}}
\newcommand{\Ai}[1]{A\sp{({#1})}}
\newcommand{\Di}[1]{D\sp{({#1})}}
\newcommand{\bi}[1]{b\sp{({#1})}}
\newcommand{\wi}[1]{w\sp{({#1})}}
\newcommand{\Fi}[1]{F\sp{({#1})}}
\newcommand{\Ideg}{\mathcal{I}_{\mathrm{in}}}
\newcommand{\Min}{\mathbf{M}_{\mathrm{in}}}
\newcommand{\Md}{\mathbf{M}_{\mathrm{in}}}

\newcommand{\Msp}{\mathbf{M}_{\mathrm{sp}}}
\newcommand{\Minv}{\mathbf{M}_v}
\newcommand{\MinV}{\mathbf{M}_V}
\newcommand{\IinV}{\mathcal{I}_V}
\newcommand{\Iinv}{\mathcal{I}_v}

\newcommand{\Isp}{\mathcal{I}_{\mathrm{sp}}}
\newcommand{\order}[1]{\mathrm{O}(#1)}
\newcommand{\X}{\mathcal{X}}
\newcommand{\bb}{$b$B}
\newcommand{\mrbb}{MR$b$B}

\usepackage{geometry}
\usepackage{newtxtext,newtxmath}

\title{The $b$-branching problem in digraphs}

\author{
Naonori Kakimura\thanks{Department of Mathematics, 
Faculty of Science and Technology, 
Keio University, 
Kanagawa 223-8522, Japan. 
\texttt{kakimura@math.keio.ac.jp}}
\and 
Naoyuki Kamiyama\thanks{Institute of Mathematics for Industry, Kyushu University, and JST, PRESTO, 
Fukuoka 819-0395, Japan.
  \texttt{kamiyama@imi.kyushu-u.ac.jp}}
\and 
Kenjiro Takazawa\thanks{Department of Industrial and Systems Engineering, Faculty of Science and Engineering, 
Hosei University, Tokyo 184-8584, Japan.  
{\tt takazawa@hosei.ac.jp}} 
}

\date{February, 2018}
\begin{document}

\maketitle

\begin{abstract}
In this paper, 
we introduce the concept of $b$-branchings in digraphs, which is 
a generalization of branchings 
serving as a counterpart of $b$-matchings. 
Here $b$ is a positive integer vector on the vertex set of a digraph, 
and 
a $b$-branching is defined as a common independent set of two matroids defined by $b$: 
an arc set is a $b$-branching if it has at most $b(v)$ arcs sharing the terminal vertex $v$, 
and  
it is an independent set of a certain sparsity matroid defined by $b$. 
We demonstrate that $b$-branchings yield an appropriate generalization of branchings 
by extending several classical results on branchings. 
We first present 
a multi-phase greedy algorithm for finding a maximum-weight $b$-branching. 
We then prove 
a packing theorem extending Edmonds' disjoint branchings theorem, 
and 
provide a strongly polynomial algorithm for finding optimal disjoint $b$-branchings. 
As a consequence of the packing theorem, 
we prove the integer decomposition property of the $b$-branching polytope. 
Finally, 
we deal with a further generalization 
in which 
a matroid constraint is imposed on the $b(v)$ arcs sharing the terminal vertex $v$.

\end{abstract}

\section{Introduction}
\label{SECintro}

Since the pioneering work of Edmonds \cite{Edm70,Edm79}, 
the importance of  
{\em matroid intersection} has been well appreciated. 
A special class of matroid intersection is 
\emph{branchings} (or \emph{arborescences}) in digraphs.  
Branchings have several good properties which do not hold for general matroid intersection. 
The objective of this paper is to propose a class of matroid intersection 
which generalizes branchings and inherits those good properties of branchings.

One of the good properties of branchings is that 
a maximum-weight branching can be found by a simple combinatorial algorithm \cite{Boc71,CL65,Edm67,Ful74}. 
This algorithm is 
much simpler than general weighted matroid intersection algorithms, 
and is 
referred to 
as a ``multi-phase greedy algorithm'' in the textbook by Kleinberg and Tardos \cite{KT05}. 

Another good property is the elegant theorem for packing disjoint branchings \cite{Edm73}. 
In terms of matroid intersection, 
this theorem says that, 
if there exist $k$ disjoint bases in each of the two matroids, 
then there exist $k$ disjoint common bases. 
This packing theorem leads to a proof that 
the branching polytope has the \emph{integer decomposition property} (defined in Section \ref{SECpre}).

In this paper, 
we propose \emph{$b$-branchings}, 
a class of matroid intersection 
generalizing branchings, 
while maintaining the above two good properties. 
This offers a new direction of fundamental extensions of the classical theorems on branchings. 

Let $D=(V,A)$ be a digraph 
and let $b\in \ZZ_{++}^V$ be a positive integer vector on $V$. 
For $v \in V$ and $F \subseteq A$, 
let $\delta^-_F(v)$ denote the set of arcs in $F$ entering $v$, 
and 
let $d\sp{-}_F(v) = |\delta^-_F(v)|$. 
One matroid $\Mdeg$ on $A$ has its independent set family $\Ideg$ defined by 
\begin{align}
\label{EQpartition}
&{}\Ideg = \{F \subseteq A \colon  \mbox{$d_F^-(v) \le b(v)$ for each $v \in V$}\}. 
\end{align}
That is, 
$\Mdeg$ is 
the direct sum of a uniform matroid on $\delta_A^-(v)$ of rank $b(v)$ for every $v \in V$. 
Hence, 
each vertex can have indegree at most $b(v)$, 
which can be more than one. 
Indeed, 
this is the reason why we refer to it as a $b$-branching,
as a counterpart of a $b$-matching.

In order to make $b$-branchings a satisfying generalization of branchings, 
the other matroid should be defined appropriately. 
Our answer is a \emph{sparsity matroid} 
determined by $D$ and $b$, 
which is defined as follows. 
For 
$F \subseteq A$ and 
$X\subseteq V$, 
let $F[X]$ denote the set of arcs in $F$ induced by $X$. 
Also, 
denote 
$\sum_{v \in X}b(v)$ by $b(X)$. 
Now define a matroid $\Msp$ on $A$  with independent set family $\Isp$ by 
\begin{align}
\label{EQsparsity}
&{}\Isp = \{F \subseteq A \colon \mbox{$|F[X]| \le b(X) - 1$ ($\emptyset \neq X \subseteq V$)}\}. 
\end{align}
It is known that $\Msp$ is a matroid \cite[Theorem 13.5.1]{Fra11}, 
referred to as a \emph{count matroid} or a \emph{sparsity matroid}.

Now we refer to an arc set $F \subseteq A$ as a \emph{b-branching} if $F \in \Ideg \cap \Isp$. 
It is clear that 
a branching is a special case of a $b$-branching where $b(v)=1$ for each $v\in V$. 
We demonstrate that 
$b$-branchings 
yield a reasonable generalization 
of branching 
by proving that the two fundamental results on branchings can be extended. 
That is, 
we present 
a multi-phase greedy algorithm for finding a maximum-weight $b$-branching, 
and 
a theorem for packing disjoint $b$-branchings.

Our multi-phase greedy algorithm is an extension of the maximum-weight branching algorithm \cite{Boc71,CL65,Edm67,Ful74}, 
and it has the following features. 
First, 
its running time is $\mathrm{O}(|V||A|)$, 
which is as fast as a simple implementation of the maximum-weight branching algorithm \cite{Boc71,CL65,Edm67,Ful74}, 
and 
faster than the current best general weighted matroid intersection algorithm. 
Second, 
our algorithm also finds an optimal dual solution, 
which is integer if the arc weights are integer. 
Thus, 
the algorithm constructively proves the total dual integrality of the associated linear inequality  system. 
Finally, 
the algorithm leads to a characterization of the existence of a $b$-branching with prescribed indegree, 
which is a generalization of 
that for an arborescence \cite{Boc71,Edm67,Ful74}. 

This characterization theorem is extended to 
a theorem on packing disjoint $b$-branchings. 
Let 
$k$ be a positive integer, 
and 
$b_1,\ldots, b_k$ be nonnegative integer vectors on $V$ such that $b_i(v)\le b(v)$ for each $v \in V$ and $b_i\neq b$ ($i =1,\ldots, k$). 
We provide a necessary and sufficient condition for $D$ to contain $k$ disjoint $b$-branchings $B_1,\ldots, B_k$ 
such that 
$d^-_{B_i}(v)=b_i(v)$ for every $v \in V$ and $i=1,\ldots, k$, 
which extends Edmonds' disjoint branching theorem \cite{Edm73}. 
We then show such disjoint $b$-branchings $B_1,\ldots, B_k$ can be found in strongly polynomial time. 
This strongly polynomial solvability is extended to finding disjoint $b$-branchings $B_1,\ldots, B_k$ 
that minimize $w(B_1)+ \cdots + w(B_k)$, when the arc-weight vector $w \in \RR_+^A$ is given. 
By utilizing our disjoint $b$-branchings theorem, 
we also prove the integer decomposition property of the $b$-branching polytope. 

We further deal with a generalized class of \emph{matroid-restricted $b$-branchings}. 
This is a class of matroid intersection 
in which $\Mdeg$ is the direct sum of an arbitrary matroid on $\delta_A^-(v)$ of rank $b(v)$ for all $v \in V$. 
Note that, 
in the class of $b$-branchings, 
the matroid $\Mdeg$ is the direct sum of a uniform matroid on $\delta_A^-(v)$ of rank $b(v)$. 
We show that our multi-phase greedy algorithm can be extended to this generalized class. 

\medskip

Let us conclude this section with describing related work.
The weighted matroid intersection problem is a common generalization of various combinatorial optimization problems such as bipartite matchings, packing spanning trees, and branchings~(or arborescences) in a digraph.
The problem has also been applied to various engineering problems, e.g., in electric circuit theory~\cite{M00,R89}, rigidity theory~\cite{R89}, and
network coding~\cite{DFZ11,HKM05}.
Since 1970s, quite a few algorithms have been proposed for matroid intersection problems, e.g., \cite{BCG86,Fra81,IT76,LSW2015,L70,L75}~(See \cite{HKK16} for further references).
However, all known algorithms are not greedy, but based on augmentation; repeatedly incrementing a current solution by exchanging some elements.

The matroids in branchings are a partition matroid and a graphic matroid, which are interconnected by a given digraph.
Such interconnection makes branchings more interesting.
As mentioned before, branchings have properties that matroid intersection of an arbitrary pair of a partition matroid and a graphic matroid does not have.
In particular, extending the packing theorem of branchings \cite{Edm73} is indeed a recent active topic.
Kamiyama, Katoh, and Takizawa \cite{KKT09} presented a fundamental extension
based on reachability in digraphs,
which is followed by a further extension based on
convexity in digraphs due to Fujishige \cite{Fuj10}.
Durand de Gevigney, Nguyen, and Szigeti \cite{DNS13} proved a theorem for
packing arborescences with matroid constraints.
Kir{\' a}ly \cite{Kir16} generalized the result of \cite{DNS13} in the
same direction of \cite{KKT09}.
A matroid-restricted packing of arborescences \cite{BK16,Fra09} is another generalization concerning a matroid constraint. 
We remark that our packing and matroid restriction for $b$-branchings differ from 
the above matroidal extensions of packing of arborescences.

\medskip

The organization of this paper is as follows. 
In Section \ref{SECpre}, 
we review the literature of branchings and matroid intersection, 
including  
algorithmic, polyhedral, and packing results. 
In Section \ref{SECalgo}, 
we present a multi-phase greedy algorithm for finding a 
maximum-weight $b$-branching. 
Section \ref{SECpacking} is devoted to 
proving a theorem on packing disjoint $b$-branchings. 
In Section \ref{SECmatroid}, 
we extend the multi-phase greedy algorithm to matroid-restricted $b$-branchings. 
In Section \ref{SECconcl}, 
we conclude this paper with a couple of remarks. 

\section{Preliminaries}
\label{SECpre}

In this section, 
we review fundamental results on branchings and related theory of matroid intersection and polyhedral combinatorics. 
For more details, 
the readers are referred to \cite{Kam14,KV12,Sch03}.

In a digraph $D=(V,A)$, 
an arc subset $B \subseteq A$ is a \emph{branching} 
if, 
in the subgraph $(V,B)$, 
the indegree of every vertex is at most one 
and 
there does not exist a cycle. 
In terms of matroid intersection, 
a branching is a common independent set 
of a partition matroid and a graphic matroid, 
i.e.,\ 
intersection of 
\begin{align}
\label{EQpartition1}
&{}\{F \subseteq A \colon  \mbox{$d^-_F(v) \le 1$ for each $v \in V$}\}, \\
\label{EQgraphic}
&{}\{F \subseteq A \colon \mbox{$|F[X]| \le |X| - 1$ ($\emptyset \neq X \subseteq V$)}\}. 
\end{align}
Recall that a branching is a special case of a $b$-branching where $b(v)=1$ for each $v\in V$. 
Indeed, 
by putting $b(v)=1$ for each $v\in V$ in \eqref{EQpartition} and \eqref{EQsparsity}, 
we obtain \eqref{EQpartition1} and \eqref{EQgraphic}, 
respectively. 

As stated in Section \ref{SECintro}, 
a maximum-weight branching can be found by a multi-phase greedy algorithm \cite{Boc71,CL65,Edm67,Ful74}, 
which appears in standard textbooks such as \cite{KT05,KV12,Sch03}. 
To the best of our knowledge, 
we have no other nontrivial special case of matroid intersection 
which can be solved greedily. 
For example, 
intersection of two partition matroids is equivalent to bipartite matching. 
This seems the simplest nontrivial example of matroid intersection, 
but 
we do not know a greedy algorithm for finding a maximum bipartite matching.

Another important result on branchings is the disjoint branchings theorem by Edmonds \cite{Edm73}, 
described as follows. 
For a positive integer $k$, 
the set of integers  
$\{1,\ldots, k\}$ is denoted by $[k]$.  
For $F \subseteq A$ and $X \subseteq V$, 
let $\delta^-_F(X) \subseteq A$ denote the set of arcs in $F$ from $V \setminus X$ to $X$, 
and let $d^-_F(X) = |\delta^-_F(X)|$. 

\begin{theorem}[Edmonds \cite{Edm73}]
\label{THMbpacking}
Let $D=(V,A)$ be a digraph and $k$ be a positive integer, 
and $U_1,\ldots, U_k$ be subsets of $V$. 
Then, 
there exist disjoint branchings $B_1,\ldots, B_k$ such that 
$U_i = \{v \in V \colon d^-_{B_i}(v)=1\}$ for each $i \in [k]$ if and only if 
\begin{align}
\notag
d^-_A(X) \ge |\{ i \in [k] \colon X \subseteq U_i\}| \quad (\emptyset \neq X \subseteq V). 
\end{align}
\end{theorem}

From Theorem \ref{THMbpacking}, 
we obtain a theorem on covering a digraph by branchings \cite{Fra79,MG86}.

\begin{theorem}[\cite{Fra79,MG86}]
\label{THMbcovering}
Let $D=(V,A)$ be a digraph and let $k$ be a nonnegative integer. 
Then, 
the arc set $A$ can be covered by $k$ branchings if and only if 
\begin{align*}
&{}d^-_A(v) \le k \quad (v \in V), \\
&{}|A[X]| \le k(|X|-1) \quad (\emptyset \neq X \subseteq V). 
\end{align*}
\end{theorem}

Theorem \ref{THMbcovering} leads to the \emph{integer decomposition property} of the branching polytope. 
The \emph{branching polytope} is a convex hull of the charactiristic vectors of all branchings. 
It follows from the total dual integrality of matroid intersection \cite{Edm70} that 
the branching polytope is determined by the following linear system: 
\begin{alignat}{2}
\label{EQmi1}
&{}x(\delta^-(v)) \le 1 \quad {}&&{}(v \in V), \\
\label{EQmi2}
&{}x(A[X]) \le |X|-1 \quad    {}&&{}(\emptyset \neq X \subseteq V), \\
\label{EQmi3}
&{} x(a) \ge 0               {}&&{}(a\in A). 
\end{alignat}
\begin{theorem}[see \cite{Sch03}]
The linear system \eqref{EQmi1}--\eqref{EQmi3} is totally dual integral. 
\end{theorem}
\begin{corollary}[see \cite{Sch03}]
\label{CORbTDI}
The linear system \eqref{EQmi1}--\eqref{EQmi3} determines the branching polytope. 
\end{corollary}

For a polytope $P$ and a positive integer $k$, 
define $kP=\{ x \colon \exists x' \in P, x=kx'\}$. 
A polytope $P$ has the \emph{integer decomposition property} 
if, 
for each positive integer $k$, 
any integer vector $x \in kP$ can be represented as the sum of $k$ integer vectors in $P$. 
The integer decomposition property of the branching polytope is a direct consequence of Theorem \ref{THMbcovering} and 
Corollary \ref{CORbTDI}. 
\begin{corollary}[\cite{BT81}]
The branching polytope has the integer decomposition property. 
\end{corollary}

We remark that 
the integer decomposition property does not hold for an arbitrary matroid intersection polytope. 
Schrijver \cite{Sch03} presents an example of matroid intersection defined on the edge set of $K_4$ without integer decomposition property. 
Indeed, 
finding a class of polyhedra with integer decomposition property is a 
classical topic in combinatorics. 
Typical examples of polyhedra with integer decomposition property include 
polymatroids \cite{BT81,Gil75}, 
the branching polytope \cite{BT81}, 
and 
intersection of two strongly base orderable matroids \cite{DM76,McD76}. 
While there is some recent progress \cite{Ben17}, 
the integer decomposition property of polyhedra is far from being well understood. 
In Section \ref{SECpacking}, 
we will prove that the $b$-branching polytope is 
a new example of polytopes with 
integer decomposition property.

\section{Multi-phase greedy algorithm}
\label{SECalgo}

In this section, 
we present a multi-phase greedy algorithm 
for finding a maximum-weight $b$-branching 
by extending that for branchings \cite{Boc71,CL65,Edm67,Ful74}.

\subsection{Key lemma}
\label{SECkeylemma}
Let $D=(V,A)$ be a digraph and $b \in \ZZ_{++}^V$ be a positive integer vector on $V$. 
Recall that an arc set $F \subseteq A$ is a $b$-branching 
if $F \in \Ideg \cap \Isp$, 
where $\Ideg$ and $\Isp$ are defined by \eqref{EQpartition} 
and \eqref{EQsparsity}, 
respectively. 

We first show a key property of $\Md$ and $\Msp$, 
which plays an important role in our algorithm. 

\begin{lemma}
\label{LEMcircuit}
An independent set $F$ in $\Mdeg$ is not independent in $\Msp$ 
if and only if $(V,F)$ has a strong component $X$ such that 
\begin{align}
\label{EQstrong}
|F[X]| = b(X). 
\end{align}
Moreover, 
such $F[X]$ is a circuit of $\Msp$.
\end{lemma}

\begin{proof}
Sufficiency is obvious: 
$|F[X]| = b(X)$ implies that $F$ is not independent in $\Msp$. 

We now prove necessity. 
Suppose that $F \in \Ideg$ is not independent in $\Msp$. 
Then, 
there exists $X$ ($\emptyset \neq X \subseteq V$) such that 
\begin{align}
\label{EQdep}
|F[X]| \ge b(X). 
\end{align}
Let $X$ be an inclusionwise minimal set satisfying \eqref{EQdep}. 
That is, 
\begin{align}
\label{EQminimal}
|F[X']| \le b(X') - 1 \quad (\emptyset \neq X' \subsetneqq X).
\end{align}

We first show that $X$ satisfies \eqref{EQstrong}. 
Since $F$ is independent in $\Mdeg$, 
it holds that 
\begin{align}
\label{EQYtight}
\mbox{$|F[Y]| \le \displaystyle\sum_{v\in Y}d_F^-(v) \le b(Y)$ for each $Y\subseteq V$}.
\end{align}  
By \eqref{EQdep} and \eqref{EQYtight}, 
it follows that 
\begin{align}
\label{EQXtight}
|F[X]| = \sum_{v \in X}d^-_F(v) = b(X), 
\end{align}
and thus  
\eqref{EQstrong} holds.

We then prove that $X$ is a strong component in $(V,F)$. 
The former equality in \eqref{EQXtight} implies that 
\begin{align}
\label{EQXin}
d^-_F(X)=\sum_{v \in X} d^-_F(v) - |F[X]| = 0. 
\end{align}
The latter equality in \eqref{EQXtight}  implies that 
\begin{align}
\label{EQvsatur}
\mbox{$d^-_F(v) = b(v)$ for every $v \in X$}. 
\end{align}
Then, 
it follows from \eqref{EQminimal} and \eqref{EQvsatur} that 
\begin{align}
\notag
d^-_F(X') 
{}&{}= \sum_{v \in X'}d^-_F(v) - |F[X']| \\
{}&{}\ge b(X') - (b(X')-1)= 1 \quad
\quad (\emptyset \neq X' \subsetneqq X). 
\label{EQXin2}
\end{align}
By 
\eqref{EQXin} and \eqref{EQXin2}, 
we have shown that $X$ is a strong component in $(V,F)$. 

We complete the proof by showing 
that $F[X]$ is a circuit in $\Msp$. 
The fact that $F[X] \notin \Isp$ directly follows from $|F[X]|=b(X)$. 
Thus, 
it is sufficient to prove that $F'=F[X] \setminus \{a\} \in \Isp$ for each arc $a \in F[X]$. 

It follows from \eqref{EQXtight} that 
\begin{align}
\label{EQX}
|F'[X]| = |F[X]| - 1 = b(X) -1. 
\end{align}
For a proper subset $X'$ of $X$, 
by \eqref{EQminimal}, 
\begin{align}
\label{EQXprime2}
|F'[X']| \le |F[X']| \le b(X')-1. 
\end{align}
By \eqref{EQX} and \eqref{EQXprime2}, 
we conclude that $F' \in \Isp$.
\end{proof}

Lemma \ref{LEMcircuit} enables us to design the following multi-phase greedy algorithm for finding a maximum-weight $b$-branching: 
\begin{itemize}
\item
Find a maximum-weight independent set $F$ in $\Mdeg$. 
\item
If $(V,F)$ has a strong component $X$ satisfying \eqref{EQdepend}, 
then contract $X$, 
reset $b$ and the weights of the remaining arcs appropriately, 
and recurse. 
\end{itemize}
A formal description of the algorithm appears in Section \ref{SECdescription}.

\subsection{Algorithm description}
\label{SECdescription}

We denote an arc $a \in A$ with initial vertex $u$ and terminal vertex $v$ by $(u,v)$. 
We assume that the arc weights are nonnegative and represented by a vector $w \in \RR_+^A$. 
For $F \subseteq A$, 
we denote $w(F)=\sum_{a \in F}w(a)$. 

Our multi-phase greedy algorithm for finding a maximum-weight $b$-branching is described as follows. 
\begin{description}
\item[\textsc{Algorithm \bb{}}.]
\item[Input.]
A digraph $D=(V,A)$, 
and 
vectors $b \in \ZZ_{++}^V$ 
and 
$w \in \RR_+^A$. 

\item[Output.]
	A $b$-branching $F \subseteq A$ maximizing $w(F)$. 
\item[Step 1.]
	Set $i:=0$, 
	$\Di{0} := D$, 
	$\bi{0} := b$,  
	and 
	$\wi{0} := w$. 
\item[Step 2.]
	Define a matroid $\Mdegi{i}=(\Ai{i}, \Idegi{i})$ 
	accordingly to $\Di{i}$ and $\bi{i}$ by \eqref{EQpartition}. 
	Then, 
	find $\Fi{i} \in \Idegi{i}$ maximizing $\wi{i}(\Fi{i})$. 
\item[Step 3.]
	If $(\Vi{i},\Fi{i})$ has a strong component $X$ such that 
	\begin{align}
	\label{EQdepend}
	|\Fi{i}[X]| = \bi{i}(X), 
	\end{align}
	then 
	go to Step 4. 
	Otherwise, 
	let $F := \Fi{i}$ and 
	go to Step 5. 
\item[Step 4.]
	Denote by $\X \subseteq 2\sp{\Vi{i}}$ the family of strong components $X$ in $(\Vi{i},\Fi{i})$ satisfying \eqref{EQdepend}. 
	Execute the following updates to construct $\Di{i+1}=(\Vi{i+1}, \Ai{i+1})$, 
	$\bi{i+1}\in \ZZ_{++}\sp{\Vi{i+1}}$, 
	and $\wi{i+1} \in \RR_+\sp{\Ai{i+1}}$. 
	\begin{itemize}
	\item
		For each $X \in \X$, 
		execute the following updates. 
		First, 
		contract $X$ to obtain a new vertex $v_X$. 
		Then, 
		for every arc $a=(z,y) \in \Ai{i}$ with $z \in \Vi{i} \setminus X$ and $y \in X$, 
		\begin{align*}
		&{}z' := 	\begin{cases}
					v_{X'} & (\mbox{$z \in X'$ for some $X' \in \X$}), \\
					z      & (\mbox{otherwise}),
					\end{cases}
					\\
		&{}a' := (z',v_X), \\
		&{}\Psi(a') := a, \\
		&{}\wi{i+1}(a') := \wi{i}(a) - \wi{i}(\alpha(a,\Fi{i})) + \wi{i}(a_X), 
		\end{align*}
		where 
		$\alpha(a,\Fi{i})$ is an arc in $\delta_{\Fi{i}}^- (y)$ minimizing $\wi{i}$, 
		and 
		$a_X$ is an arc in $\Fi{i}[X]$ minimizing $\wi{i}$. 
	\item
		Define $\bi{i+1}\in \ZZ_{++}\sp{\Vi{i+1}}$ by
		\begin{align*}
		\bi{i+1}(v) := 	\begin{cases}
							1    &(\mbox{$v = v_X$ for some $X \in \X$}),\\
							\bi{i}(v) &(\mbox{otherwise}). 
						\end{cases}
		\end{align*}
	\end{itemize}
	Let $i := i+1$ and go back to Step 2. 

\item[Step 5.]
	If $i=0$, then return $F$. 
\item[Step 6.]
	For every strong component $X$ in $(\Vi{i-1}, \Fi{i-1})$ with \eqref{EQdepend}, 
	apply the following update: 
	if there exists $a' = (z,v_X) \in F$, 
	then
	\begin{align*}
	F:= ((F \setminus \{a'\}) \cup \{\Psi(a')\}) \cup (\Fi{i-1}[X] \setminus \{\alpha(\Psi(a'),X')\}); 
	\end{align*}
	otherwise, 
	\begin{align*}
	F:= F \cup (\Fi{i-1}[X] \setminus \{a_X\}). 
	\end{align*}
	Let $i:= i-1$ and go back to Step 5. 
\end{description}

The complexity of Algorithm \bb{} is analyzed as follows. 
It is clear that there are at most $|V|$ iterations. 
It is also straightforward to see that the $i$-th iteration requires $\order{|\Ai{i}|}$ time: 
Steps 2, 3, and 4 respectively require $\order{|\Ai{i}|}$ time. 
Thus, 
the total time complexity of the algorithms is $\order{|V||A|}$.

\subsection{Optimality of the algorithm and totally dual integral system}

In this subsection,  
we prove that the output of \textsc{Algorithm \bb{}} is a maximum-weight 
$b$-branching by the following primal-dual argument. 
We first present a linear program describing the maximum-weight $b$-branching problem. 
It is a special case of the linear program for weighted matroid intersection, 
and hence 
we already know that 
the linear system is endowed with total dual integrality. 
Here we show an algorithmic proof for the total dual integrality. 
That is, 
we show that, 
when $w$ is an integer vector, 
integral optimal primal and dual solutions can be computed via 
\textsc{Algorithm \bb}. 

Consider the following linear program, 
in variable $x \in \RR\sp{A}$, 
associated with the maximum-weight $b$-matching problem: 
\begin{alignat}{3}
\label{EQlp0}
&{}\mbox{maximize} \quad {}&&{}\sum_{a \in A}w(a)x(a) &&\\
\label{EQlp1}
&{}\mbox{subject to}\quad{}&&{}x(\delta_A^-(v)) \le b(v) \quad {}&&{}(v \in V), \\
\label{EQlp2}
&&&{}x(A[X]) \le b(X) - 1 \quad    {}&&{}(\emptyset \neq X \subseteq V), \\
\label{EQlp3}
&&&{}0\le x(a) \le 1               {}&&{}(a\in A). 
\end{alignat}
The constraints \eqref{EQlp1}--\eqref{EQlp3} are indeed a special case of 
a linear system describing the common independent sets in two matroids, 
which is totally dual integral (see \cite{Sch03}). 
Thus, 
we obtain the following theorem. 

\begin{theorem}
\label{THMpolytope}
The linear system \eqref{EQlp1}--\eqref{EQlp3} is totally dual integral. 
In particular, 
the linear system \eqref{EQlp1}--\eqref{EQlp3} determines the $b$-branching polytope. 
\end{theorem}

The dual problem of \eqref{EQlp0}--\eqref{EQlp3}, 
in variable $p \in \RR\sp{2^V}$ and $q \in \RR\sp{A}$, 
is described as follows. 

\begin{alignat}{2}
\label{EQdual0}
&{}\mbox{minimize} \quad {}&&{}\sum_{v \in V}b(v)p(v) + \sum_{X \colon \emptyset \neq X \subseteq V}(b(X) - 1)p(X) + \sum_{a\in A}q(a)\\
&{}\mbox{subject to} \quad {}&&{} p(v) + \sum_{X\colon a \in A[X]}p(X) + q(a) \ge w(a) \quad (a=uv \in A), \\
&&&{}p(X) \ge 0 \quad (X \subseteq V), \\
&&&{}q(a)  \ge 0 \quad (a \in A).
\label{EQdual2}
\end{alignat}

An optimal solution $(p^*,q^*)$ is computed via \textsc{Algorithm \bb} in the following manner. 
At the beginning of \textsc{Algorithm \bb}, 
set 
$w^\circ=w$. 
In Step 4 of \textsc{Algorithm \bb}, 
for each strong component $X \in \X$, 
define $p^*(X) \in \RR$ by 
\begin{align*}
\notag
p^*(X) = \min
\{
	{}&{}\min\{w^\circ(\alpha^\circ(a)) - w^\circ(a) \colon a \in \delta_{\Ai{i}}^-(X)\} , 
	\min\{w^{\circ}(a') \colon a' \in \Fi{i}[X] \}
\}, 
\end{align*}
where 
$\alpha^\circ(a)$ is the $b(y)$-th optimal arc 
with respect to $w^\circ$ 
among the arcs sharing the terminal vertex $y \in V$ with $a$ in the original digraph $D$. 
Then 
for each arc $a \in A$ such that 
$a \in \Ai{i}[X]$ or $a$ is deleted in the contraction of $X'$ with $v_{X'}$ included in $X$, 
set 
$w^\circ(a) :=w^\circ(a) - p^*(X)$.  
After the termination of \textsc{Algorithm \bb}, 
let 
the value $p^*(v)$ be equal to the $b(v)$-th maximum  value among $\{w^\circ(a)\colon a \in \delta_A^-(v)\}$ 
for each vertex $v \in V$. 
Finally, 
let $q^*(a) = \max \{w(a) - p^*(v) - \sum_{X \colon a \in A[X]}p^*(X),0\}$. 

It is straightforward to see that 
the characteristic vector of the output $F$ and $(p^*,q^*)$  
satisfy the complementary slackness condition. 
Thus 
they are 
optimal solutions for the linear programs \eqref{EQlp0}--\eqref{EQlp3} and \eqref{EQdual0}--\eqref{EQdual2}, 
respectively. 
Moreover, 
$(p^*,q^*)$ is integer if $w$ is integer, 
which implies that 
\eqref{EQlp1}--\eqref{EQlp3} is totally dual integral.

\subsection{Existence of a $b$-branching with prescribed indegree}

Our algorithm leads to the following theorem characterizing the existence of 
$b$-branching with prescribed indegree, 
which is an extension of 
that for arborescences. 

\begin{theorem}
\label{THMarb}
Let $D=(V,A)$ be a digraph 
and 
$b\in \ZZ_{++}\sp{v}$ be a positive integer vector on $V$. 
Let $b' \in \ZZ_{+}^V$ be a nonnegative integer vector such that 
$b'(v) \le b(v)$ for every $v \in V$ and 
$b' \neq b$. 
Then, 
$D$ has a $b$-branchings $B$ such that 
$d^-_{B}(v) = b'(v)$ for each $v \in V$
if and only if 
\begin{alignat}{2}
\label{EQexistdeg}
{}&{}d^-_A(v) \ge b'(v)\quad {}&&{}(v \in V),\\
\label{EQexistcut}
{}&{}d^-_A(X) \ge 1 {}&&{}(\mbox{$\emptyset \neq X \subsetneq V$, $b'(X) = b(X) \neq 0$}). 
\end{alignat}
\end{theorem}

Let $r \in V$ be a specified vertex. 
A characterization of the existence of an $r$-arborescence \cite{Boc71,Edm67,Ful74} is obtained 
as a special case of Theorem \ref{THMarb}, 
by putting 
$b(v)=1$ for every $v \in V$, 
$b'(v)=1$ for every $v \in V \setminus \{r\}$, 
and $b'(r)=0$. 

Theorem \ref{THMarb} can be proved in two ways. 
The necessity of \eqref{EQexistdeg} and \eqref{EQexistcut} is clear. 
One way to derive the sufficiency of \eqref{EQexistdeg} and \eqref{EQexistcut} is 
Algorithm \bb. 
Apply Algorithm \bb{} to the case where $b=b'$ and $w(a)=1$ for each $a \in A$. 
Then, 
\eqref{EQexistdeg} and \eqref{EQexistcut} certify that 
$F^{(i)}$ found in Step 2 of Algorithm \bb{} is always a base of $\Min\sp{(i)}$. 
It thus follows that 
the output $F$ of Algorithm \bb{} is a $b$-branching with $d_F^-=b'$. 

An alternative proof for the sufficiency of \eqref{EQexistdeg} and \eqref{EQexistcut} 
is implied by  
the proof for Theorem \ref{THMbbpacking} in Section \ref{SECpacking}, 
which extends Theorem \ref{THMarb}  to 
a characterization of the existence of disjoint $b$-branchings with prescribed indegree.

\section{Packing disjoint $b$-branchings}
\label{SECpacking}

In this section, 
we present a theorem on packing disjoint $b$-branchings $B_1,\ldots, B_k$ with 
prescribed indegree, 
which extends Theorem \ref{THMbpacking}, as well as Theorem \ref{THMarb}. 
Our proof is an extension of the proof for Theorem \ref{THMbpacking} by Lov\'{a}sz \cite{Lov76}. 
We then show that 
such disjoint $b$-branchings can be found in strongly polynomial time. 
We further show that disjoint $b$-branchings $B_1, \ldots, B_k$ minimizing the weight $w(B_1)+ \cdots + w(B_k)$ 
can be found in strongly polynomial time. 
Finally, 
as a consequence of our packing theorem, 
we prove the integer decomposition property of the $b$-branching polytope. 

\subsection{Characterizing theorem for disjoint $b$-branchings}

Let $D=(V,A)$ be a digraph, 
$b\in \ZZ_{++}^V$ be a positive integer vector on $V$, 
and 
$k$ be a positive integer. 
For $i\in [k]$, 
let $b_i \in \ZZ_{+}^V$ be a nonnegative integer vector such that 
$b_i(v) \le b(v)$ for every $v \in V$ and 
$b_i \neq b$. 
We present a theorem for chracterizing whether $D$ contains disjoint $b$-branchings 
$B_1,\ldots, B_k$ such that 
$d^-_{B_i} = b_i$
for each $i \in [k]$.  

We begin with introducing 
a function which plays a key role in the sequel. 
Define a function $g: 2^V \to \ZZ_+$ by
\begin{align}
\label{EQg}
g(X) = |\{ i \in [k] \colon b_i(X) = b(X) \neq 0\}| \quad (X \subseteq V). 
\end{align}
The following lemma is straightforward to observe. 
\begin{lemma}
\label{LEMgsup}
The function $g$ is supermodular. 
\end{lemma}

\begin{proof}
For $X \subseteq V$, 
define $I_X \subseteq [k]$ by 
$I_X = \{ i \in [k] \colon b_i(X) = b(X) \neq 0\}$. 
By the definition \eqref{EQg} of $g$, 
for $X,Y \subseteq V$, 
it holds that 
\begin{align*}
g(X) + g(Y) = |I_X| + |I_Y| = |I_X \setminus I_Y| + |I_Y \setminus I_X| + 2|I_X \cap I_Y|. 
\end{align*}
Moreover, 
it is straightforward to see that 
\begin{align*}
&{}g(X \cup Y) = |I_X \cap I_Y|, &
&{}g(X \cap Y) \ge 
|I_X \cup I_Y| = |I_X \setminus I_Y| + |I_Y \setminus I_X| + |I_X \cap I_Y|. 
\end{align*}
It therefore holds that 
$g(X)+g(Y) \le g(X \cup Y) + g(X \cap Y)$, 
and 
hence $g$ is a supermodular function. 
\end{proof}

Our characterization theorem is described as follows. 
\begin{theorem}
\label{THMbbpacking}
Let $D=(V,A)$ be a digraph, 
$b\in \ZZ_{++}^V$ be a positive integer vector on $V$, 
and 
$k$ be a positive integer. 
For $i\in [k]$, 
let $b_i \in \ZZ_{+}^V$ be a nonnegative integer vector such that 
$b_i(v) \le b(v)$ for every $v \in V$ and 
$b_i \neq b$. 
Then, 
$D$ has disjoint $b$-branchings $B_1,\ldots, B_k$ such that 
$d^-_{B_i} = b_i$
for each $i \in [k]$ 
if and only if 
the following two conditions are satisfied: 
\begin{alignat}{2}
\label{EQpackingdeg}
{}&{}d^-_A(v) \ge \sum_{i=1}^k b_i(v)\quad {}&&{}(v \in V),\\
\label{EQpackingcut}
{}&{}d^-_A(X) \ge g(X) \quad{}&&{} (X \subseteq V). 
\end{alignat}
\end{theorem}

\begin{proof}
Necessity is clear. 
We prove sufficiency by induction on $\sum_{i=1}^k b_i(V)$. 
The case $\sum_{i=1}^k b_i(V) = 0$ is trivial;
$B_i = \emptyset$ for each $i \in [k]$. 

Without loss of generality, 
suppose that $b_1(V) >0$. 
Define a partition $\{V_0, V_1, V_2\}$ of $V$ by 
\begin{align*}
&{}V_0 = \{ v \in V \colon b_1(v) = 0\}, \\
&{}V_1 = \{ v \in V \colon 0<b_1(v) < b(v)\}, \\
&{}V_2 = \{ v \in V \colon b_1(v) = b(v)\}.
\end{align*}
Then, 
it holds that 
\begin{align}
\label{EQpartB}
&{}V_0 \cup V_1 \neq \emptyset, \\
\label{EQpartR}
&{}V_0 \neq V, 
\end{align}
which follow from $b_1 \neq b$ and $b_1(V)>0$, respectively.

Let $W \subseteq V$ be an inclusionwise minimal vertex subset satisfying 
\begin{align}
\label{EQw1}
&{}W \cap (V_0\cup V_1) \neq \emptyset, \\
\label{EQw2}
&{}W \setminus V_0 \neq \emptyset, \\
\label{EQw3}
&{}d^-_A(W)=g(W).
\end{align}
Such $W \subseteq V$ always exists, 
because $W=V$ satisfies \eqref{EQw1}--\eqref{EQw3}: 
\eqref{EQw1} follows from \eqref{EQpartB}; 
\eqref{EQw2} follows from \eqref{EQpartR}; 
and 
\eqref{EQw3} follows from $b_i\neq b$ ($i \in [k]$) and hence $g(V)=0$.  
Let $W_j = W \cap V_j$ ($j=0,1,2$). 

\begin{claim}
\label{CLa}
There exists an arc $(u,v) \in A$ such that $u \in W_0\cup W_1$ and $v \in W_1 \cup W_2$.
\end{claim}

\begin{proof}
First, 
suppose that $W_2 \neq \emptyset$. 
Then, 
it holds that $g(W_2)>g(W)$, 
since 
every $i \in [k]$ contributing to $g(W)$ also contributes to $g(W_2)$, 
and 
$i=1$ does not contribute to $g(W)$ but to $g(W_2)$. 
Hence we obtain 
that
\begin{align}
d^-_A(W_2) \ge  g(W_2) 
> g(W) = d^-_A(W). 
\label{EQw22}
\end{align}
Now \eqref{EQw22}
implies that 
there exists an arc $(u,v) \in A$ such that $u \in W_0\cup W_1$ and $v \in W_2$. 

Next, 
suppose that $W_2 = \emptyset$. 
By \eqref{EQw2}, 
we have that $W_1 \neq \emptyset$. 
Then, 
it holds that 
\begin{align*}
\sum_{v \in W_1}d^-_A(v) 
&{}\ge \sum_{i=1}^k b_i(W_1) \quad (\because \mbox{\eqref{EQpackingdeg}}) \\
&{}>   \sum_{i=2}^k b_i(W_1) \quad (\because b_1(W_1) >0) \\
&{}\ge |\{i \in [k] \colon b_i(W) = b(W) \neq 0\}| \quad (\because b_1(W) \neq b(W))\\
&{}= g(W) \\
&{}= d^-_A(W), 
\end{align*}
implying that 
there exists an arc $(u,v) \in A$ such that $u \in W=W_0 \cup W_1$ and $v \in W_1$. 
\end{proof}

Let $a = (u,v) \in A$ be an arc in Claim \ref{CLa}. 
We then show that 
resetting 
\begin{align}
\label{EQresetA}
&{}A :=  A \setminus \{a\}, \\
\label{EQresetB}
&{}b_1(v) := b_1(v) - 1 
\end{align}
maintains \eqref{EQpackingdeg} and \eqref{EQpackingcut}. 
(This resetting amounts to augmenting $B_1$ by adding $a$.)

It is straightforward to see that \eqref{EQresetA} and \eqref{EQresetB} maintain 
\eqref{EQpackingdeg}. 
To prove that \eqref{EQresetA} and \eqref{EQresetB} maintain 
\eqref{EQpackingcut}, 
suppose to the contrary that 
$X \subseteq V$ comes to violate \eqref{EQpackingcut} after the resetting \eqref{EQresetA} and \eqref{EQresetB}. 

This violation implies that $d^-_A (X) = g(X)$ before the resetting, and 
$d^-_A(X)$ has decreased by one 
while $g(X)$ has remained unchanged by the resetting. 
It then follows that 
\begin{align}
\label{EQenter}
u \in V  \setminus X \quad \mbox{and} \quad v \in X.
\end{align} 
It also follows that $i=1$ does not contribute to $g(X)$, 
and hence before the resetting, 
it holds that 
\begin{align}
\label{EQU}
X \cap (V_0 \cup V_1) \neq \emptyset.
\end{align}  

By \eqref{EQenter}, 
we have that 
$u \in  W\setminus X$ and 
$v\in X \cap W$, 
and hence $\emptyset \neq X \cap W  \subsetneqq W$. 
Here 
we show that 
$X \cap W$ satisfies 
\eqref{EQw1}--\eqref{EQw3}, 
which contradicts the minimality of $W$. 

Before the resetting, 
it holds that 
\begin{align}
\label{EQsubsup1}
d^-_A(X \cap W) 
&{}\le 
d^-_A(X) + d^-_A(W) - d^-_A(X \cup W) 
\\
\label{EQsubsup2}
&{} \le g(X) + g(W) - g(X \cup W) 
\quad  \\
\label{EQsubsup3}
&{} \le g(X \cap W).
\end{align}
Indeed, 
\eqref{EQsubsup1} follows from 
submodularity of $d^-_A$. 
The inequality \eqref{EQsubsup2} follows from $d^-_A(X)=g(X)$, $d^-_A(W)=g(W)$, 
and $d^-_A(X \cup W) \ge g(X \cup W)$. 
Finally, 
\eqref{EQsubsup3} follows from Lemma \ref{LEMgsup}. 

Since 
$d^-_A(X \cap W) \ge g(X\cap W)$ by \eqref{EQpackingcut}, 
all inequalities \eqref{EQsubsup1}--\eqref{EQsubsup3} 
hold with equality, 
and hence
$d^-_A(X \cap W) 
= 
g(X \cap W)
$
holds before the resetting. 

Equality in \eqref{EQsubsup3} implies that 
$(X \cap W) \cap (V_0 \cup V_1) \neq \emptyset$. 
Indeed, 
we have that 
$W \cap (V_0 \cup V_1) \neq \emptyset$ because $u \in W \cap (V_0 \cup V_1)$, 
and hence $i=1$ does not contribute to $g(W)$.  
Combined with \eqref{EQU}, 
$i=1$ contributes to none of $g(X)$, $g(W)$, and $g(X \cup W)$. 
Thus, 
by the equality in \eqref{EQsubsup3}, 
$i=1$ does not contribute to $g(X \cap W)$ as well, 
and hence $(X \cap W) \cap (V_0 \cup V_1) \neq \emptyset$ must hold. 

We also have $(X \cap W)\setminus V_0 \neq \emptyset$, 
because $v \in (X \cap W)\setminus V_0$. 
Therefore, 
$X \cap W$ satisfies \eqref{EQw1}--\eqref{EQw3}, 
contradicting the minimality of $W$. 
Thus, 
we have finished proving that 
resetting of \eqref{EQresetA} and \eqref{EQresetB} maintains 
\eqref{EQpackingcut}. 

Now we can apply induction to obtain 
disjoint $b$-branchings $B_1,\ldots, B_k$ in the digraph $(V, A \setminus\{a\})$ 
such that 
$d^-_{B_1} = b_1 - \chi_v$ and $d^-_{B_i} = b_i$ for $i=2,\ldots, k$, 
where $\chi_v \in \ZZ\sp{V}$ is a vector defined by $\chi_v(v) = 1$ and $\chi_v(u) = 0$ for every $u \in V \setminus \{v\}$. 
We complete the proof by showing that $B_1 \cup \{a\}$ is a $b$-branching. 

In resetting, 
we always have $u \in W_0 \cup W_1$, 
which implies that 
the construction of $B_1$ begins with a vertex $r$ with $b_1(r) < b(r)$ 
and 
the component in $(V, B_1)$ containing $a$ includes $r$. 
Thus, 
no $X \subseteq V$ comes to satisfy $|B_1[X]| = b(X)$. 
\end{proof}

\subsection{Algorithm for finding disjoint $b$-branchings}

Let us discuss the algorithmic aspect of Theorem \ref{THMbbpacking}. 
First, 
we can determine whether \eqref{EQpackingdeg} and \eqref{EQpackingcut} hold in strongly polynomial time. 
Condition \eqref{EQpackingdeg} is clear. 
For \eqref{EQpackingcut}, 
we have that $d^-_A(X)$ is submodular and 
$g(X)$ 
is supermodular (Lemma \ref{LEMgsup}), 
and hence  
$d^-_A(X) - g(X)$ is submodular. 
Thus, 
we can determine whether there exists $X$ with $d^-_A(X) - g(X) < 0$ by 
submodular function minimization, 
which can be done in strongly polynomial time \cite{IFF01,LSW2015,Sch00}. 

Finding $b$-branchings $B_1,\ldots, B_k$ can also be done in strongly polynomial time. 
By the proof for Theorem \ref{THMbbpacking}, 
it suffices to find an arc $a \in A$ such that 
resetting \eqref{EQresetA} and \eqref{EQresetB} maintains \eqref{EQpackingcut}. 
This can be done by determining whether there exists $X$ with $d^-_A(X) - g(X) < 0$ 
after resetting \eqref{EQresetA} and \eqref{EQresetB} for each $a \in A$, 
i.e.,\ 
at most $|A|$ times of submodular function minimization \cite{IFF01,LSW2015,Sch00}. 

\begin{theorem}
Conditions \eqref{EQpackingdeg} and \eqref{EQpackingcut} can be checked in strongly polynomial time. 
Moreover, 
if \eqref{EQpackingdeg} and \eqref{EQpackingcut} hold, 
then disjoint $b$-branchings $B_1,\ldots, B_k$ such that $d_{B_i}^- = b_i$ for each $i \in [k]$ can be found in strongly polynomial time. 
\end{theorem}

Furthermore, 
if an arc-weight vector $w \in \RR_+^A$ is given, 
we can find disjoint $b$-branchings $B_1,\ldots, B_k$ minimizing 
$w(B_1)+ \cdots +w(B_k)$ in strongly polynomial time. 
Indeed, 
conditions \eqref{EQpackingdeg} and \eqref{EQpackingcut} derive a totally dual integral system 
which determines a \emph{submodular flow polyhedron}. 
A set family $\C \subseteq 2^V$ is called a 
\emph{crossing family} if,  
for each $X,Y \in \C$ with
$X \cup Y \neq V$ 
and $X \cap Y \neq \emptyset$, 
it holds that 
$X \cup Y, X\cap Y \in \C$. 
A function $f: \C \to \RR$ defined on a 
crossing family $\C \subseteq V$ is called 
\emph{crossing submodular} if, 
for each 
$X,Y \in \C$ with 
$X \cup Y \neq V$ 
and 
$X \cap Y \neq \emptyset$, 
it holds that 
$f(X) + f(Y) \ge f(X\cup Y) + f(X\cap Y)$. 
A function $f$ is \emph{crossing supermodular} if $-f$ is crossing submodular. 
A \emph{submodular flow polyhedron} is a polyhedron described as 
\begin{alignat*}{2}
&{}x(\delta_A^-(X))-x(\delta_A^+(X)) \le f(X)  {}&\quad&{}
(X \in \C)
, \\
&{}l(a) \le x(a) \le u(a) {}&{}
 \quad &{}(a \in A)
\end{alignat*}
by some digraph $(V,A)$, 
crossing 
submodular function $f$ 
on a crossing family $\C \subseteq 2^V$, 
and 
vectors $l,u \in \RR^A$, 
where $\delta_A^+(X)$ denotes the set of arcs 
in $A$ from $X$ to $V \setminus X$.

\begin{lemma}[\cite{Sch84}]
\label{LEMsf}
For a digraph $D=(V,A)$, 
let 
$f\colon 2^V \to \RR$ be a 
crossing 
supermodular function 
on $\C \subseteq 2^V$ 
and $u \in \RR^A$. 
Then, 
a polyhedron determined by 
\begin{alignat*}{2}
&x(\delta^-_A (X)) \ge f(X) \quad &&{}
(X \in \C),\\
& 0 \le x(a) \le u(a) \quad &&{}(a \in A)
\end{alignat*}
is a submodular flow polyhedron.
\end{lemma}

By Lemma \ref{LEMsf}, the linear inequality system \eqref{EQpackingdeg} and \eqref{EQpackingcut} determines a submodular flow polyhedron.
Indeed, we can define a crossing supermodular function $f\colon 2^V \to \RR$ by
\[
f(X) =
\begin{cases}
\displaystyle \sum_{i=1}^k b_i (v) & \mbox{($X = \{v\}$ for some $v\in V$)},\\
g(X) & (\mbox{otherwise}). 
\end{cases}
\]
Since a submodular flow polyherdron is totally dual integral \cite{EG77},
an arc subset $B \subseteq A$ with \eqref{EQpackingdeg} and \eqref{EQpackingcut} 
minimizing $w(B)$ can be found by optimization over a 
submodular flow polyhedron, 
which can be done in strongly polynomial time \cite{FIM02,FT87,IMS00,IMS03}. 
After that, 
we can partition $B$ into $b$-branchings $B_1,\ldots, B_k$ with $d^-_{B_i}=b_i$ ($i \in [k]$) 
in the same manner as above. 

\begin{theorem}
If \eqref{EQpackingdeg} and \eqref{EQpackingcut} hold, 
then disjoint $b$-branchings $B_1,\ldots, B_k$ such that $d_{B_i}^- = b_i$ for each $i \in [k]$ 
minimizing $w(B_1)+ \cdots +w(B_k)$ can be found in strongly polynomial time. 
\end{theorem}

\subsection{Integer decomposition property of the $b$-branching polytope}

In this subsection 
we show another consequence of 
Theorem \ref{THMbbpacking}:  
the integer decomposition property of the $b$-branching polytope. 
First, 
Theorem \ref{THMbbpacking} 
leads to the following min-max relation on covering by $b$-branchings. 
This is an extension of Theorem \ref{THMbcovering}, 
the theorem on covering by branchings \cite{Fra79,MG86}. 

\begin{corollary}
\label{CORbbcover}
Let $D=(V,A)$ be a digraph, 
$b\in \ZZ_{++}^V$ be a positive integer vector on $V$, 
and 
$k$ be a positive integer. 
Then, 
the arc set $A$ can be covered by $k$ $b$-branchings if and only if 
\begin{align}
\label{EQcover1}
&{}d^-_A(v) \le k \cdot b(v) \quad (v \in V), \\
\label{EQcover2}
&{}|A[X]| \le k (b(X)-1) \quad (\emptyset \neq X \subseteq V). 
\end{align}
\end{corollary}

\begin{proof}
Necessity is obvious. 
To prove sufficiency, 
construct a new digraph $D'=(V',A')$ in the following manner. 
The vertex set $V'$ is obtained from $V$ by adding a new vertex $r$. 
The arc set $A'$ is obtained from $A$ by adding $k\cdot b(v) - d^-_A(v)$ parallel arcs 
from $r$ to $v$ for each $v \in V$. 
Note that $k\cdot b(v) - d^-_A(v)$ is nonnegative by \eqref{EQcover1}. 

Then, 
in the digraph $D' = (V',A')$, 
it holds that 
\begin{align}
\label{EQcover1p}
d^-_{A'}(v) {}&{}= k \cdot b(v) \quad (v \in V), \\
d^-_{A'}(X) 
{}&{}= \sum_{v \in X}d^-_{A'}(v) - |A[X]| \notag\\
\label{EQcover2p}
{}&{}\ge  \sum_{v \in X}k \cdot b(v) - k(b(X)-1) = k \quad (\emptyset \neq X \subseteq V).  
\end{align}
Now define vectors $b', b_0' \in \ZZ^{A'}$ by 
\begin{align*}
&
b'(v) = 
\begin{cases}
b(v) & (v \in V), \\
1 & (v =r), 
\end{cases}
&&
b'_0(v) = 
\begin{cases}
b(v) & (v \in V), \\
0 & (v =r).
\end{cases}
\end{align*}
By \eqref{EQcover1p} and \eqref{EQcover2p}, 
we can apply Theorem \ref{THMbbpacking} to $D'$ and obtain 
$k$ disjoint $b'$-branchings $B_1',\ldots B_k'$ in $D'$ 
satisfying $d^-_{B_i'} = b_0'$ for each $i \in [k]$. 
It then follows that $|B_i'| = b_0'(V)=b(V)$ for each $i \in [k]$. 
Since $$|A'|= |A| + \sum_{v \in V} (k \cdot b(v)-d^-_A(v)) = |A| + (k\cdot b(V) - |A|)=k \cdot b(V),$$
$\{B_1',\ldots , B_k'\}$ is a partition of $A'$. 
Thus, 
by 
restricting $B_1',\ldots B_k'$ to $A$, 
we obtain 
$b$-branchings $B_1,\ldots, B_k$ partitioning $A$. 
\end{proof}

The integer decomposition property of the $b$-branching polytope 
is a direct consequence of 
Corollary \ref{CORbbcover}.  

\begin{corollary}
The $b$-branching polytope has the integer decomposition property. 
\end{corollary}

\begin{proof}
Denote the $b$-branching polytope by $P$. 
Recall that $P$ is determined by 
\eqref{EQlp1}--\eqref{EQlp3} (Theorem \ref{THMpolytope}). 
Let $k$ be a positive integer and 
$x \in \ZZ\sp{A}$  be an integer vector in $kP$. 
It follows from $x \in kP$ that 
\begin{alignat*}{2}
{}&{}x (\delta^-(v)) \le k\cdot b(v) \quad {}&&{}(v \in V), \\
{}&{}x(A[X]) \le k(b(X) - 1) \quad {}&&{}(\emptyset \neq X \subseteq V), \\
{}&{}0 \le x(a) \le k \quad {}&&{}(a \in A). 
\end{alignat*}
Now consider an arc set $A_x$ consisting of 
$x(a)$ arcs parallel to $a$ for each $a \in A$. 
It is straightforward to see that 
\eqref{EQcover1} and \eqref{EQcover2} hold when $A=A_x$. 
Thus, 
by Corollary \ref{CORbbcover}, 
$A_x$ can be covered by $k$ $b$-branchings. 
In other words, 
$x$ is the sum of $k$ integer vectors in $P$, 
implying the integer decomposition property of $P$. 
\end{proof}

\section{Matroid-restricted $b$-branchings}
\label{SECmatroid}

In this section, 
we deal with \emph{matroid-restricted $b$-branching}, 
which further generalizes $b$-branchings. 
Let $D=(V,A)$ be a digraph and $b \in \ZZ_{++}^V$ be a positive integer vector on $V$. 
For each vertex $v \in V$, 
a matroid $\Minv=(\delta^-(v), \Iinv)$ with rank $b(v)$ is attached. 
We denote the direct sum of $\Minv$ for every $v \in V$ by $\MinV = (A, \IinV)$. 
Now an arc set $F \subseteq A$ is an \emph{$\MinV$-restricted $b$-branching} 
if $F \in \IinV \cap \Isp$. 
Note that a $b$-branching is a special case where 
$\Minv$ is a uniform matroid for each $v\in V$. 

Here we provide a multi-phase greedy algorithm for finding a maximum-weight $\MinV$-restricted $b$-branching 
by extending \textsc{Algorithm \bb}. 

\begin{description}
\item[Algorithm \mrbb{}.]
\item[Input.]
A digraph $D=(V,A)$, 
vectors $b :\ZZ_{++}^V$ and 
$w\in \RR_+^A$, 
and 
matroids $\Minv=(\delta^-(v), \Iinv)$ with rank $b(v)$ for each $v\in V$. 

\item[Output.]
	An $\MinV$-restricted $b$-branching $F \subseteq A$ maximizing $w(F)$. 
\item[Step 1.]
	Set $i:=0$, 
	$\Di{0} := D$, 
	$\bi{0} = b$,  
	and 
	$\wi{0} := w$. 
\item[Step 2.]
	For each $v \in \Vi{i}$, 
	define a matroid $\Mvi{i}=(\delta_{\Ai{i+1}}^-(v), \Ivi{i+1})$ as 
	$\Mvi{i}$ if $v \in V$, 
	and 
	a uniform matroid of rank $1$ otherwise.  
	Let $\MVi{i} = (\Ai{i}, \IVi{i})$ be the direct sum of $\Mvi{i}$ for every $v \in \Vi{i}$. 
	Then, 
	find $\Fi{i} \in \IVi{i}$ maximizing $\wi{i}(\Fi{i})$. 
\item[Step 3.]
	If $(\Vi{i},\Fi{i})$ has a strong component $X$ such that 
	\begin{align}
	\label{EQdependM}
	|\Fi{i}[X]| = \bi{i}(X), 
	\end{align}
	then 
	go to Step 4. 
	Otherwise, 
	let $F := \Fi{i}$ and 
	go to Step 5. 
\item[Step 4.]
	Denote the family of strong components $X$ in $(\Vi{i},\Fi{i})$ satisfying \eqref{EQdependM} by $\X \subseteq 2\sp{\Vi{i}}$. 
	Execute the following updates to construct $\Di{i+1}=(\Vi{i+1}, \Ai{i+1})$, 
	$\bi{i+1}\in  \ZZ_{++}\sp{\Vi{i+1}}$, 
	and $\wi{i+1}\in\RR_+\sp{\Ai{i+1}}$. 
	\begin{itemize}
	\item
		For each $X \in \X$, 
		execute the following updates. 
		First, 
		contract $X$ to obtain a new vertex $v_X$. 
		Then, 
		for every arc $a=(z,y) \in \Ai{i}$ with $z \in \Vi{i} \setminus X$ and $y \in X$, 
		\begin{align*}
		&{}z' := 	\begin{cases}
					v_{X'} & (\mbox{$z \in X'$ for some $X' \in \X$}), \\
					z      & (\mbox{otherwise}),
					\end{cases}
					\\
		&{}a' := (z',v_X), \\
		&{}\Psi(a') := a, \\
		&{}\wi{i+1}(a') := \wi{i}(a) - \wi{i}(\alpha(a,\Fi{i})) + \wi{i}(a_X), 
		\end{align*}
		where 
		$\alpha(a,\Fi{i})$ is an arc in 
		the fundamental circuit of $a$ with respect to $\Fi{i}$ in $\mathbf{M}_y^{(i)}$
		minimizing $\wi{i}$, 
		and 
		$a_X$ is an arc in $\Fi{i}[X]$ minimizing $\wi{i}$. 
	\item
		Define $\bi{i+1} \in  \ZZ_{++}\sp{\Vi{i+1}}$ by
		\begin{align*}
		\bi{i+1}(v) := 	\begin{cases}
							1    &(\mbox{$v = v_X$ for some $X \in \X$}),\\
							\bi{i}(v) &(\mbox{otherwise}). 
						\end{cases}
		\end{align*}
	\end{itemize}
	Let $i := i+1$ and go back to Step 2. 

\item[Step 5.]
	If $i=0$, then return $F$. 
\item[Step 6.]
	For every strong component $X$ in $(\Vi{i-1}, \Fi{i-1})$ with \eqref{EQdependM}, 
	apply the following update: 
	if there exists $a' = (z,v_X) \in F$, 
	then
	\begin{align*}
	F:= ((F \setminus \{a'\}) \cup \{\Psi(a')\}) \cup (\Fi{i-1}[X] \setminus \{\alpha(\Psi(a'),X')\}); 
	\end{align*}
	otherwise, 
	\begin{align*}
	F:= F \cup (\Fi{i-1}[X] \setminus \{a_X\}). 
	\end{align*}
	Let $i:= i-1$ and go back to Step 5. 
\end{description}

\section{Concluding remarks}
\label{SECconcl}

In this paper, 
we have proposed $b$-branchings, 
a generalization of branchings. 
In a $b$-branching, a vertex $v$ can have indegree at most $b(v)$, 
and thus $b$-branchings serve as a counterpart of $b$-matchings for matchings. 

It is somewhat surprising that, 
to the best of our knowledge, 
such a fundamental generalization of branchings has never 
appeared in the literature. 
The reason might be that, 
in order to obtain a reasonable generalization, 
it is far from being trivial how the other matroid (graphic matroid) in branchings is generalized. 
We have succeeded in obtaining a generalization 
inheriting the multi-phase greedy algorithm \cite{Boc71,CL65,Edm67,Ful74} 
and the packing theorem \cite{Edm73} for branchings 
by setting 
a sparsity matroid defined by \eqref{EQsparsity} as 
the other matroid.

An important property of the two matroids is 
Lemma \ref{LEMcircuit}, 
which says that 
an independent set of one matroid is decomposed into 
an independent set and 
some circuits in the other matroid. 
This plays an important role in the design of a multi-phase greedy algorithm: 
find an optimal independent $F$ set in one matroid; 
contract the circuits in $F$ with respect to the other matroid; 
and 
the optimal common independent set can be found recursively. 
We remark that the definitions \eqref{EQpartition} and \eqref{EQsparsity} are essential 
to attain this property. 
For example, 
the property fails if 
the vector $b$ is not identical in \eqref{EQpartition} and \eqref{EQsparsity}. 
It also fails if 
the sparsity matroid is defined by $|F[X]| \le b(X) - k$ for $k \neq 1$. 

Another remark is on
the similarity of our algorithm and the blossom algorithm for nonbipartite matchings \cite{Edm65}, 
where a factor-critical component can be contracted and expanded. 
In our $b$-branching algorithm, 
for each strong component $X \in \X$ and each $v^*\in X$, 
there exists an arc set $F_X \subseteq A[X]$ such that 
$d^-_{F_X} (v^*) = b(v^*)-1$ and $d^-_{F_X} (v^*) = b(v^*)$ for each $v \in X \setminus \{v^*\}$. 
In the blossom algorithm for nonbipatite matchings, 
for each factor-critical component $X$ and each vertex $v^* \in X$, 
there exists a matching exactly covering $X \setminus \{v^*\}$.

We finally remark that 
the problem of finding a maximum-weight $b$-branching is a special case of a 
modest generalization of the framework of the $\mathcal{U}$-feasible $t$-matching problem in bipartite graphs~\cite{Tak17ipco}. 
In \cite{Tak17ipco}, 
it is proved that 
the $\mathcal{U}$-feasible $t$-matching problem 
in bipartite graphs is efficiently tractable under certain assumptions on the family of excluded structures $\mathcal{U}$. 
The $b$-branching problem can be regarded as a new problem 
which falls in this tractable class of the (generalized) $\mathcal{U}$-feasible $t$-matching problem.

\section*{Acknowledgements}
This work is partially supported by 
JST ERATO Grant Number JPMJER1201,  
JST CREST Grant Number JPMJCR1402,  
JST PRESTO Grant Number JPMJPR14E1, 
JSPS KAKENHI Grant Numbers 
JP16K16012, 
JP17K00028, 
JP25280004, 
JP26280001, 
Japan.

\end{document}